\documentclass[9pt,twoside]{pnas-new}

\articletype{CLASSIFICATION}

\templatetype{pnasmathematics} 

\csappto{ps@firststyle}{\fancyhead{}} 

\usepackage{amsthm}
\usepackage{bm}
\usepackage{graphicx}
\newcommand{\vexp}{\alpha_{\text{vac}}}
\newcommand{\mexp}{\alpha_{\text{mag}}}

\usepackage{tikz}
\usetikzlibrary{decorations.pathmorphing}
\usetikzlibrary{arrows.meta,positioning,calc}

\newtheorem{theorem}{Theorem}
\newtheorem{lemma}{Lemma}

\newtheorem{remark}{Remark}

\newtheorem{definition}{Definition}
\newtheorem{assumption}{Assumption}

\begin{document}

\title{Elasticized magnetohydrostatics}

\author[a,b,2]{J. W. Burby}
\author[a,b]{I. Cho}
\author[a,b]{M. Ruth}

\affil[a]{Department of Physics, The University of Texas at Austin, Austin, TX 78712, USA}
\affil[b]{Institute for Fusion Studies, The University of Texas at Austin, Austin, TX 78712, USA}

\leadauthor{Burby}

\significancestatement{The design of stellarator fusion reactors requires solving an optimization problem in a very high dimensional space. Analytical theory is meant to serve as a light that guides optimizers through this formidable design space. But it partially fails in this respect due to foundational disagreements between theory and computation and incomplete coverage of high-pressure operation regimes. This work proposes a new framework for modeling steady-state stellarator configurations that repairs the disconnect between theory and computation, and unlocks high-pressure reactor design.}

\authorcontributions{Please provide details of author contributions here.}
\authordeclaration{Please declare any competing interests here.}
\correspondingauthor{\textsuperscript{2}To whom correspondence should be addressed. E-mail: joshua.burby\@austin.utexas.edu}


\begin{abstract}
The modern approach to modeling equilibria in stellarators features two blind spots: lack of foundational agreement between theory and computation, and inability to reliably explore economically-attractive high-$\beta$ concepts. These blind spots originate at a technical level from lack of a useful existence theory for the partial differential equations believed to describe equilibria, and from the encumbrance of crude, inflexible models for the topology of the confining magnetic field. This Article proposes the first stellarator equilibrium model that rigorously guarantees existence of an equilibrium with any chosen magnetic topology. It completely separates determination of the magnetic topology from the equilibrium existence question, and thereby enables importing realistic tearing physics into equlibrium computations. The new model replaces the standard plasma potential energy functional with its ensemble average over rapid kinetic-scale fluctuations. Averaging ``elasticizes" the potential by energetically penalizing separation of field lines. The elasticized potential's super-linear dependence on plasma strain embodies a phenomenological model for profile stiffness, and plays an essential role in the existence theory.

\end{abstract}

\dates{This manuscript was compiled on \today}
\doi{\url{www.pnas.org/cgi/doi/10.1073/pnas.XXXXXXXXXX}}

\maketitle
\thispagestyle{firststyle}
\ifthenelse{\boolean{shortarticle}}{\ifthenelse{\boolean{singlecolumn}}{\abscontentformatted}{\abscontent}}{}

\Firstpage


Stellarator fusion is advancing rapidly, driven by the success of Wendelstein 7-X \cite{klingerOverviewFirstWendelstein2019} and growing private investment \cite{schmittMagnetohydrodynamicEquilibriumStability2025,prostEconomicallyOptimizedDesign2024,lionStellarisHighfieldQuasiisodynamic2025,swansonOverviewHeliosDesign2025}. A central challenge for a stellarator fusion pilot plant is the ability to reliably predict and optimize three-dimensional stellarator equilibria, including the magnetic field $\bm{B}$ and plasma pressure $p$. However, the current approach to equilibrium modeling features two blind spots: lack of foundational agreement between theory and computation, and inability to reliably explore economically-attractive high-$\beta$ concepts. Here $\beta$ denotes the ratio of plasma pressure to magnetic energy density -- a number that quantifies how efficiently the magnetic field confines the plasma fuel. These blind spots originate at a technical level from lack of a useful existence theory for the magnetohydrostatics (MHS) equations, $(\nabla\times\bm{B})\times\bm{B} = \nabla p$, $\nabla\cdot\bm{B}=0$, and from the encumbrance of crude, inflexible models for the topology of the confining magnetic field. These blind spots limit the ability of theory to guide design optimization, which is essential for development of a viable fusion pilot plant.

\textbf{Theory contradicts computation ---} Notable disagreements between theory and computation touch upon various foundational questions in stellarator confinement. Assuming MHS, Garren-Boozer \cite{garrenExistenceQuasihelicallySymmetric1991} argued quasisymmetric (QS) stellarator equilibria cannot exist. Using VMEC \cite{hirshmanSteepestdescentMomentMethod1983},  Landreman-Paul \cite{landremanMagneticFieldsPrecise2022} identified numerical equilibria with precise QS in the entire plasma volume. Assuming MHS, Cary-Shasharina \cite{caryOmnigenityQuasihelicityHelical1997,cary_helical_1997} show analytic omnigeneity (OM) implies QS. Using VMEC, Goodman \emph{et al.} \cite{goodman_constructing_2023} find vacuum fields (automatically analytic!) with precise OM and significant deviations from QS. More broadly, the theory of tearing modes predicts smooth current density across rational flux surfaces, while equilibrium computations routinely reveal singular current densities in the same locations \cite{lazersonVerificationIdealMagnetohydrodynamic2016a}. 

Disagreements between theory and computation of stellarator equilibria plausibly originate from lack of a useful MHS existence theory. In the 1960s Harold Grad recognized severe mathematical difficulties posed by these equations in non-symmetric plasma containers. He conjectured \cite{gradToroidalContainmentPlasma1967,constantin_flexibility_2021,cardona_asymmetry_2025} non-existence of smooth solutions with nested flux surfaces. Recent counterexamples \cite{landreman_analytic_2026,issan_analytic_2026} show that Grad's conjecture is false. However, a more practical weakening of Grad's conjecture still stands. It asserts smooth MHS solutions fail to exist for generic plasma boundary shape and profile functions. A significant body of computational work \cite{lazersonVerificationIdealMagnetohydrodynamic2016a,mikhailov_elimination_2019} supports the weakened conjecture. Moreover, the recent counterexamples fall far short of addressing the weakened conjecture because they only include finite-dimensional families of plasma boundary shapes. In spite of the plausible truth of a weakened Grad conjecture, most of stellarator confinement theory continues to assume MHS is actually satisfied in a stellarator. This creates a potential for vacuous theoretical claims and predictions.  On the other hand, codes like VMEC do not actually solve the MHS equations. Instead they solve implicitly-defined resolution-dependent regularizations of the MHS model. This sets up ideal conditions for vacuous theoretical predictions manifesting as contradictions between theory and computation.

The mathematical difficulties identified by Grad may not implicate the MHS model \emph{per se}, but rather the assumption of nested flux surfaces. There is some rigorous mathematical support for this perspective. Bruno-Laurence \cite{brunoExistenceThreedimensionalToroidal1996} found nearly-axisymmetric MHS solutions comprising nested toroidal annula in which the magnetic field is Beltrami. Their solutions suffer jump discontinuities across toroidal interfaces, but generally contain magnetic islands and chaos \cite{caryNoncanonicalHamiltonianMechanics1983}. The equilibrium code SPEC \cite{hudsonComputationMultiregionRelaxed2012} computes this sort of equilibrium in domains far from axisymmetry, where the Bruno-Laurence theory does not apply. More recently, Constantin-Pasqualotto \cite{constantinMagneticRelaxationVoigt2023} showed that a particular form of magnetic relaxation, similar in spirit to methods implemented in some iterative equilibrium solvers \cite{huangComputationMagnetohydrodynamicEquilibria2025,brownForcedReconnectionVoigtRegularized2026,suzukiDevelopmentApplicationHINT22006}, produces non-trivial MHS solutions in any domain. While all of these results are suggestive, they fail to provide a useful MHS existence theory for stellarators. The Bruno-Laurence solutions have only been rigorously established in nearly-axisymmetric domains, and numerical evidence \cite{quNonexistenceSteppedpressureEquilibria2021} suggests they need not exist in strongly non-axisymmetric domains. While the Constantin-Pasqualotto solutions do exist in stellarator-relevant domains, their suitability for stellarator confinement remains shrouded in mystery. The existence theory does not show these configurations support a sufficient volume of flux surfaces to realistically achieve burning plasma conditions.

\textbf{Elusive high-$\beta$ designs ---}  
Higher field strengths translate to larger magnetic coils, which translates to much higher construction costs. High-$\beta$ stellarator designs therefore generally make more efficient use of their applied magnetic field. But most design optimization studies use fast nested-surface equilibrium solvers like VMEC \cite{hirshmanSteepestdescentMomentMethod1983} or DESC \cite{paniciDESCStellaratorCode2023,dudtDESCStellaratorEquilibrium2020} inside of the optimization loop. These solvers are doomed to fail sufficiently close to the so-called equilibrium $\beta$-limit \cite{loizuEquilibriumBetalimitsClassical2017,baillodEquilibriumBetalimitsDependence2023}, where flux surfaces degrade without triggering an ideal MHD instability. Conditioned on the assumption of piecewise Beltrami fields, the high impact of allowing for non-nested magnetic topology during design optimization with $\beta > 0$ was demonstrated directly \cite{smiet_efficient_2025} through use of the SPEC in place of VMEC or DESC. If the piecewise-Beltrami assumption was known to accurately reproduce the true tearing physics \cite{connorTearingModesToroidal1988,hegnaPhysicsNeoclassicalMagnetohydrodynamic1998,hegnaPlasmaFlowHealing2012,fitzpatrickTearingModeDynamics2023} that governs magnetic topology in stellarators, these results would indicate SPEC offers the key to high-$\beta$ design. However, no existing stellarator equilibrium model, including SPEC, faithfully reproduces tearing physics due to pervasive use of crude magnetic topology models. High-$\beta$ designs remain elusive.




The crude nature of current approaches to equilibrium magnetic topology modeling is reflected in the most widespread simulation tools used for stellarator equilibrium modeling. The codes VMEC \cite{hirshmanSteepestdescentMomentMethod1983} and DESC \cite{dudtDESCStellaratorEquilibrium2020,paniciDESCStellaratorCode2023} assume the plasma-region magnetic field is tangent to concentric toroidal flux surfaces. This assumption is reasonable in low-$\beta$ configurations, but is known to fail at larger $\beta$ \cite{reiman_pressure-induced_2007,loizuEquilibriumBetalimitsClassical2017,baillodEquilibriumBetalimitsDependence2023,zhangIdentificationMHDEquilibrium2026}.
Starting in the late 1980s, the potential for non-nested magnetic topologies in stellarators motivated development of non-nested 3D equilibrium solvers. Iterative solvers, such as PIES \cite{reimanCalculationThreedimensionalMHD1986,drevlakPIESFreeBoundary2005}, SIESTA \cite{hirshmanSIESTAScalableIterative2011}, HINT2 \cite{suzukiDevelopmentApplicationHINT22006}, and MRX \cite{blickhanMRXDifferentiable3D2025} introduce physically-motivated relaxation effects that numerically generate magnetic islands and chaos, usually starting from a nested initial guess. Other solvers like SPEC \cite{hudsonComputationMultiregionRelaxed2012} and SPECTRE \cite{balkovicSPECTRERobustSolver2026} target weak solutions of MHS by extending the rigorous Bruno-Laurence construction \cite{brunoExistenceThreedimensionalToroidal1996} to strongly-shaped domains \cite{quNonexistenceSteppedpressureEquilibria2021}. Iterative schemes do not actually simulate the tearing physics responsible for island formation and saturation. Instead they introduce numerical expedients with questionable physics fidelity. In solvers based on Bruno-Laurence theory, the field in between toroidal interfaces is necessarily Beltrami. This leads to a correct topology model only when all ideal constraints except helicity between interfaces are removed. Potential issues with these crude modeling choices are only compounded by their central, inflexible role in the underlying simulation algorithms. For example, users may not freely specify the embedded topology model in any of these codes, even if a better model is available.

\textbf{A new approach ---} This Article proposes the first rigorous energy minimization principle for stellarator equilibria with user-specified magnetic topology. Due to its close connections with the theory of hyperelasticity, this equilibrium modeling framework will be referred to as elasticized magnetohydrostatics. 

The Article begins by appealing to physics beyond ideal MHD to address a major technical hurdle in minimizing the standard plasma potential energy due to Newcomb \cite{newcombLagrangianHamiltonianMethods1961}: lack of coercivity. 
Section \ref{sec:effective_mag} argues that ``elasticizing" Newcomb's magnetic energy by including the mean energy of kinetic-scale magnetic fluctuations improves the functional's growth properties. This argument extends those used to formulate a related equilibrium model in Ref.\,\cite{ruthStatisticalEquilibriumModel2026}. Section \ref{sec:vac_barrier} then argues further that combining the elasticized magnetic energy with a modified thermodynamic internal energy eliminates all apparent obstacles to coercivity. The modified internal energy differs from the usual ideal gas internal energy by a soft vacuum barrier that energetically penalizes excessive rarefaction. Based on the observations from Sections \ref{sec:effective_mag} and \ref{sec:vac_barrier}, Section \ref{sec:model_formulation} presents a precise mathematical formulation of elasticized magnetohydrostatics that defines equilibria as minimizers of a modified potential energy functional. As with Newcomb's variational principle, this model fixes the underlying magnetic topology by specifying a reference magnetic field. Unlike Newcomb's model, the potential energy functional is the sum of the modified internal energy from Section \ref{sec:vac_barrier} and the elasticized magnetic energy from Section \ref{sec:effective_mag}. Refining the heuristic arguments from Sections \ref{sec:effective_mag} and \ref{sec:vac_barrier}, Section \ref{sec:coercivity} rigorously proves coercivity of the elasticized potential energy functional. Leveraging the coercivity property and analytical techniques developed in the context of elasticity theory, Section \ref{sec:existence_theory} applies the direct method in the calculus of variations to prove existence of minimizers for the modified potential in the Sobolev space $W^{1,p}$, comprising $L^p$ back-to-labels maps \cite{cotter_multisymplectic_2007} with gradients that lie in $L^p$. Here the exponent $p$ is set by the nonlinear response of magnetic fluctuations to large-scale plasma deformations. Notably, Theorem \ref{minimizer_existence} provides a minimum potential energy state for any specified magnetic topology whatsoever, including nested, islanded, and chaotic configurations. We are not aware of any comparable results in the literature. 

By allowing a user complete freedom to specify equilibrium magnetic topology, as illustrated in Fig. \ref{fig:side_by_side}, elasticized magnetohydrostatics successfully separates magnetic topology modeling from the equilibrium existence problem. It is therefore the first stellarator equilibrium model that provides a direct interface to sophisticated magnetic topology models based on real tearing physics. It also comes equipped with the first useful existence theory for equilibria in the sense that inputs to their existence theorem directly describe essential features of the solution.

\begin{figure}
    \centering
    \includegraphics[width=0.9\linewidth]{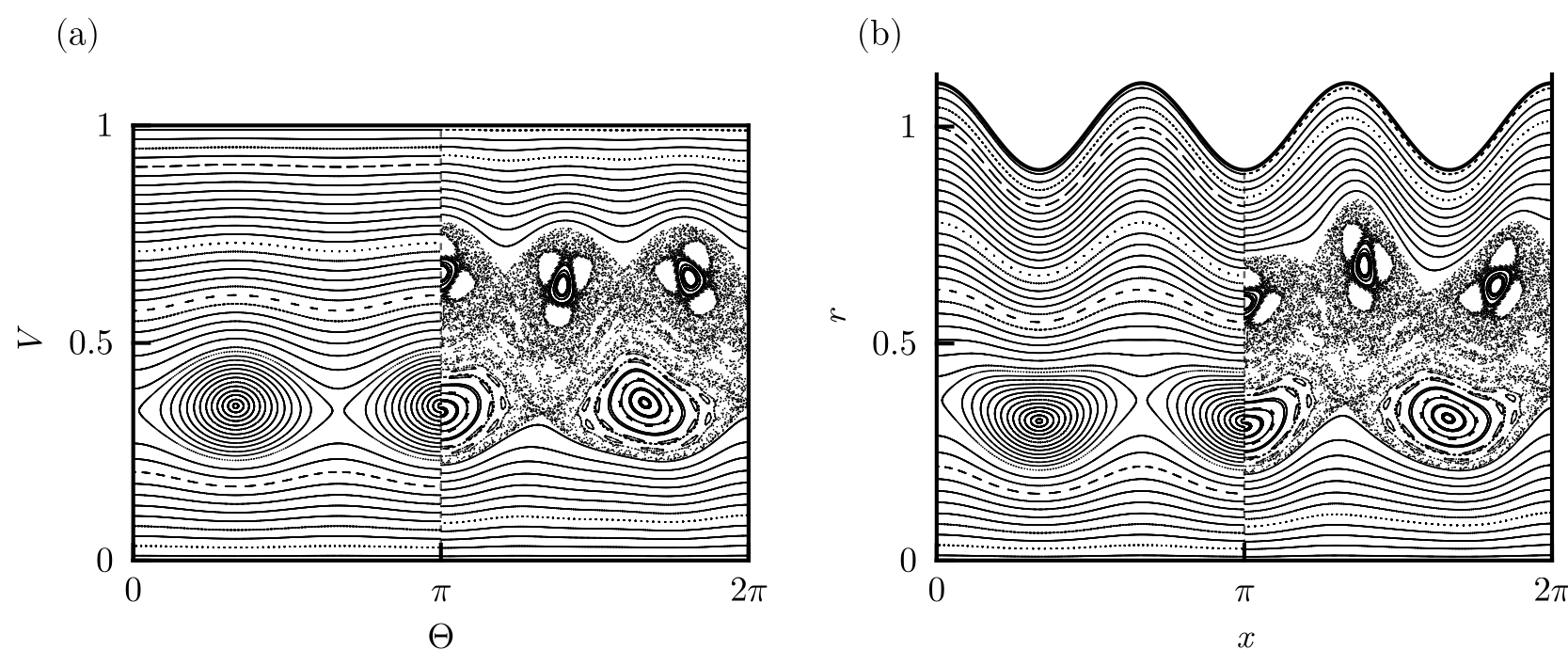}
    \caption{Reference (a) and computed (b) field-line Poincar\'e map using elasticized magnetohydrostatics. Two simulations are compared left-to-right in each figure, differing only by the user-specified magnetic topology. The left simulation features a single $3:1$ integrable island chain, while the right features overlapping $3:1$ and $5:2$ island chains. Euler-Lagrange equation residuals are $1.4\times 10^{-7}$ in each case. }
    \label{fig:side_by_side}
\end{figure}

\section{Elasticized magnetic energy\label{sec:effective_mag}}

In Newcomb's variational formulation of ideal MHD \cite{newcombLagrangianHamiltonianMethods1961}, the plasma configuration space comprises diffeomorphisms $G:Q\rightarrow Q_0$, where $Q$ denotes the plasma container and $Q_0$ denotes a reference plasma container. The diffeomorphism $G$ is known as a back-to-labels map because its value $G(\bm{x})$, $\bm{x}\in Q$, denotes the location of the fluid parcel currently located at $\bm{x}$ if the entire plasma were moved into its reference configuration. Newcomb's potential energy $\mathcal{W}$ is therefore a functional of $G$, whose specific form is given by
\begin{gather}
    \mathcal{W}(G) = \mathcal{W}_{\text{mag}}(G) + \mathcal{W}_{\text{int}}(G)\nonumber\\
    \mathcal{W}_{\text{mag}}(G)  = \frac{1}{2}\mu_0^{-1}\int_Q |\bm{B}|^2\,d^3\bm{x},\quad \mathcal{W}_{\text{int}}(G)  = \int_Q \mathcal{U}(\rho,s)\,\rho\,d^3\bm{x}\label{eq:newcomb_mag_energy},
\end{gather}
where the magnetic field, mass density, and entropy are given in terms of the back-to-labels map $G$ according to
\begin{gather*}
    \bm{B}(\bm{x}) = \mathbb{A}(\bm{x})\bm{B}_0(G(\bm{x})),\quad \rho(\bm{x}) = \mathcal{J}(\bm{x})\,\rho_0(G(\bm{x})),\quad s(\bm{x}) = s_0(G(\bm{x})),\\
    \mathbb{G}(\bm{x}) = DG(\bm{x}),\quad \mathbb{A}(\bm{x}) = \text{adj}\,\mathbb{G}(\bm{x}),\quad \mathcal{J}(\bm{x}) = \text{det}\,\mathbb{G}(\bm{x}).
\end{gather*}
Here $DG(\bm{x})$ denotes the $3\times 3$ matrix whose columns are the partial derivatives of $G$ at $\bm{x}$, $\text{adj}$ denotes the matrix adjugate, and $\bm{B}_0$, $\rho_0$, $s_0$ denote the magnetic field, the mass density, and the entropy in the plasma's reference configuration. The fact that $\bm{B}_0$, $\rho_0$, and $s_0$ are treated as fixed fields in Newcomb's potential energy functional reflects imposition of frozen-in magnetic flux, mass density, and entropy, respectively. Since the formula $\bm{B}(\bm{x}) = \mathbb{A}(\bm{x})\bm{B}_0(G(\bm{x}))$ indicates the magnetic flux $2$-form in $Q$ is the pullback of the magnetic flux $2$-form in $Q_0$, all topological features of the magnetic field $\bm{B}$ are encoded in $\bm{B}_0$.

Newcomb's potential energy functional is mathematically problematic because it fails to be coercive, which is a property that many rigorous applications of the calculus of variations rely upon \cite{dacorognaDirectMethodsCalculus2008}. The functional $\mathcal{W}(G)$ would be coercive if $\mathcal{W}(G)\rightarrow\infty$ whenever $\|G \|\rightarrow\infty$, where $\|\cdot\|$ denotes a suitable norm on the function space containing $G$. This function space and its associated norm will be described in detail in Section \ref{sec:model_formulation}. For now it is enough to assume $\|\cdot\|$ measures both values of $G$ and its first derivative $\mathbb{G}$. 

Lack of coercivity for Newcomb's functional can be understood heuristically as follows. Consider a $1$-parameter family of back-to-labels maps $G_\lambda$, with parameter $\lambda\in\mathbb{R}$. Suppose that on a subregion of $Q$ with volume $V(\lambda)\propto \lambda^{-1}$ the derivative matrix is given by 
\begin{align*}
    \mathbb{G}_\lambda = \begin{pmatrix} \lambda &0&0\\ 0& \lambda^{-1/2} & 0\\ 0&0&\lambda^{-1/2}\end{pmatrix},
\end{align*}
while elsewhere  $\mathbb{G}_\lambda = O(1)$. Now choose any $p>1$. The $p$-norm of the derivative matrix scales with $\lambda$ according to
\begin{align*}
\int_Q \| \mathbb{G}_\lambda\|^p\,d^3\bm{x}\propto \lambda^{p-1}.
\end{align*}
In particular, the $p$-norm of the gradient tends to $\infty$ as $\lambda\rightarrow\infty$. Since $\text{det}\,\mathbb{G}_\lambda = 1$ the internal energy $\mathcal{W}_{\text{int}}(G_\lambda)$ is bounded as $\lambda\rightarrow\infty$. Coercivity of Newcomb's functional would therefore require $\mathcal{W}_{\text{mag}}(G_\lambda)\rightarrow\infty$ as $\lambda\rightarrow \infty$. To assess whether this actually happens, first note that the adjugate matrix in the region where $\mathbb{G}$ is large is given by
\begin{align*}
    \mathbb{A}_\lambda = \begin{pmatrix} \lambda^{-1} &0&0\\ 0& \lambda^{1/2} & 0\\ 0&0&\lambda^{1/2}\end{pmatrix}.
\end{align*}
By the quadratic dependence of $\mathcal{W}_{\text{mag}}(G_\lambda)$ on $\mathbb{A}_\lambda$ the magnetic energy  therefore grows no faster than a constant as $\lambda\rightarrow\infty$. This shows that Newcomb's potential cannot control any $p$-norm of the derivative matrix $\mathbb{G}$ with $p > 1$, a particularly dramatic failure of coercivity. The one-parameter family $G_\lambda$ is an \textbf{escape sequence}. (Note that there is no need to consider $p < 1$ because then $(\int_Q|\cdot|^p)^{1/p}$ does not define a norm.) 


Lack of coercivity is not a necessary condition for existence of minimizers for $\mathcal{W}$, as demonstrated by the single-variable function $f(x) = -1/(1+x^2)$. So it is not fair to infer at this stage that Newcomb's potential energy functional lacks minimizers in some suitable function space. At the same time, since ideal MHD omits a great deal of important plasma physics \cite{freidbergIdealMHD2014} it is also not necessary to settle for Newcomb's potential energy as the only plausible variational description of 3D equilibria. There may be room for improving the coercivity properties of Newcomb's potential energy by including additional physics effects. 

Newcomb's form of the magnetic energy can be improved by recognizing that real stellarators sustain rapid electromagnetic turbulent fluctuations \cite{mishchenkoGlobalGyrokineticSimulations2023} that contribute to the plasma stress tensor. Here ``rapid" refers to a timescale $\tau$ much shorter than the equilibrium evolution timescale. Building on Ref.\,\cite{ruthStatisticalEquilibriumModel2026}, the argument proceeds as follows. Fluctuations of the true magnetic field $\bm{B}$ generally break the frozen-in flux condition. The reference magnetic field \[\bm{B}_0(\bm{x}_0) = \frac{\mathbb{G}(G^{-1}(\bm{x}_0))\bm{B}(G^{-1}(\bm{x}_0))}{\mathcal{J}(G^{-1}(\bm{x}_0))},\]
which would be time-independent in ideal MHD, therefore varies rapidly in time. Suppose averages over the fast time scale can be replaced with ensemble averages defined by some probability measure $P(d\bm{B}_0;G)$. Since changes in the equilibrium generally alter the statistics of gyrokinetic turbulence, this probability measure must depend parametrically on the back-to-labels map $G$. The time average of the magnetic energy then simplifies according to
\begin{align}
    \frac{1}{\tau}\int_0^\tau\,\mathcal{W}_{\text{mag}}(G)\,dt &= \mathbb{E}\left[\frac{1}{2}\mu_0^{-1}\int_Q |\mathbb{A}(\bm{x})\bm{B}_0(G(\bm{x}))|^2\,d^3\bm{x}\right]\nonumber\\
    &= \frac{1}{2}\mu_0^{-1}\int_Q |\mathbb{A}(\bm{x})\overline{\bm{\mathsf{B}}}_0(G(\bm{x}))|^2\,d^3\bm{x}\nonumber\\
    &+ \frac{1}{2}\mu_0^{-1}\int \text{tr}\,(\mathbb{A}^T(\bm{x})\mathbb{A}(\bm{x})\overline{\Sigma}_0(G(\bm{x})))\,d^3\bm{x}\label{eq:mean_mag_en}\\
    &\equiv \mathsf{W}_{\text{mag}}(G),\nonumber
\end{align}
where the one-point mean and variance of the reference magnetic field are given by
\begin{align*}
    \overline{\bm{\mathsf{B}}}_0 = \int \bm{B}_0\,P(d\bm{B}_0;G),\quad \overline{\Sigma}_0 = \int (\bm{B}_0 - \overline{\bm{\mathsf{B}}}_0)(\bm{B}_0 - \overline{\bm{\mathsf{B}}}_0)^T\,P(d\bm{B}_0;G).
\end{align*}
These expressions entail functional integration over the space of reference magnetic fields. Replacing $\mathcal{W}_{\text{mag}}(G)$ with $\mathsf{W}_{\text{mag}}(G)$ improves on Newcomb's magnetic energy by incorporating the effects of nonideal electromagnetic fluctuations. 

The first term in  \,\eqref{eq:mean_mag_en} represents the magnetic energy stored in the mean magnetic field. The second term gives the excess magnetic energy associated with fluctuations about the mean. Both the energy of the mean field and the excess energy generally exhibit a complicated $G$-dependence due to the presence of expectations over the probability measure $P(d\bm{B}_0;G)$. Deducing the precise form of this $G$-dependence from first principles likely lies beyond the reach of analytical methods. However, there is a plausible functional form motivated by simple phenomenological arguments.

In a stellarator the amplitude of turbulent magnetic fluctuations is small compared to the mean field. It is therefore reasonable to expect that the mean magnetic energy reduces to the usual ideal MHD form in the limit of vanishing fluctuations. Thus, to leading order in the fluctuation strength the mean reference field $\overline{\bm{\mathsf{B}}}_0$ is independent of $G$. Denote this $G$-independent reference field $\overline{\bm{\mathsf{B}}}_0 = \bm{\mathsf{B}}_0$. Higher-order corrections to the mean reference field correspond to fluctuation-driven dynamo effects \cite{spenceFluctuationdrivenMagneticFields2008,maoFluctuationinducedDynamoEffect2023}. This work ignores these effects. These considerations completely determine the functional form of the energy of the mean field.

\setlength{\fboxsep}{6pt}
\setlength{\fboxrule}{0.8pt}

\begin{figure}[ht]
\centering

\fbox{%
\begin{minipage}{0.42\textwidth}
\centering

\resizebox{\linewidth}{!}{%


\begin{tikzpicture}[
  scale=0.3,
  transform shape,
  spring/.style={
    draw=purple!80!black,
    line width=1.5pt,
    decorate,
    decoration={
      coil,
      segment length=1.5mm,
      amplitude=1.0mm,
      aspect=0.35
    }
  }
]

\foreach \i in {0,...,10}{
  \pgfmathsetmacro{\x}{2*\i}

  \foreach \j in {0,...,6}{
    \pgfmathsetmacro{\y}{-1.6*\j}
    \fill (\x,\y) circle (5pt);
  }

  \foreach \j in {0,...,5}{
    \pgfmathsetmacro{\yA}{-1.6*\j}
    \pgfmathsetmacro{\yB}{-1.6*(\j+1)}

    \draw[spring]
      (\x,\yA)
      --
      (\x,\yB);
  }
}

\end{tikzpicture}


}

\end{minipage}
}%
\hspace{.5em}
\fbox{%
\begin{minipage}{0.42\textwidth}
\centering

\resizebox{\linewidth}{!}{%


\begin{tikzpicture}[
  scale=0.3,
  transform shape,
  vspring/.style={
    draw=purple!80!black,
    line width=1.2pt,
    decorate,
    decoration={
      coil,
      segment length=1.4mm,
      amplitude=.9mm,
      aspect=0.35
    }
  },
  hspring/.style={
    draw=red!80!black,
    line width=1.2pt,
    decorate,
    decoration={
      coil,
      segment length=1.0mm,
      amplitude=.1mm,
      aspect=0.35
    }
  }
]

\foreach \i in {0,...,9}{
  \foreach \j in {0,...,6}{
    \pgfmathsetmacro{\xA}{2*\i}
    \pgfmathsetmacro{\xB}{2*(\i+1)}
    \pgfmathsetmacro{\y}{-1.6*\j}

    \draw[hspring]
      (\xA,\y)
      --
      (\xB,\y);
  }
}

\foreach \i in {0,...,10}{
  \foreach \j in {0,...,5}{
    \pgfmathsetmacro{\x}{2*\i}
    \pgfmathsetmacro{\yA}{-1.6*\j}
    \pgfmathsetmacro{\yB}{-1.6*(\j+1)}

    \draw[vspring]
      (\x,\yA)
      --
      (\x,\yB);
  }
}

\foreach \i in {0,...,10}{
  \foreach \j in {0,...,6}{
    \pgfmathsetmacro{\x}{2*\i}
    \pgfmathsetmacro{\y}{-1.6*\j}
    \fill (\x,\y) circle (5pt);
  }
}

\end{tikzpicture}


}

\end{minipage}
}
\caption{(Left) Instantaneously, magnetic energy can be viewed as elastic potential energy stored in field lines, shown as unconnected lines of springs. (Right) Averaging over nonideal turbulent fluctuations produces an effective, undirected elasticity effect, shown schematically as additional springs connecting neighboring field lines. \label{fig:effective_elasticity}}
\end{figure}
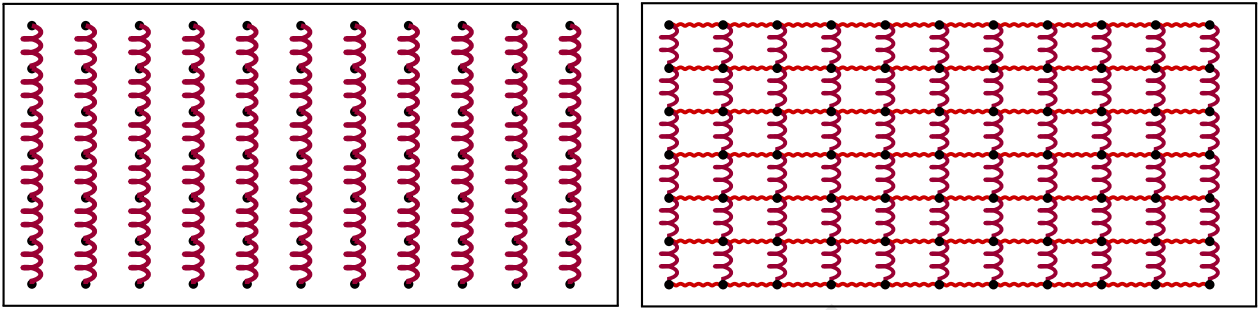

Next consider the excess energy. The instantaneous magnetic energy is proportional to the plasma strain along the reference magnetic field; components of the strain transverse to the reference field make no energetic contribution. See the left panel in Fig.\,\ref{fig:effective_elasticity}. In ideal MHD, where the reference field does not change, the plasma therefore behaves as a directed elastic medium. But as the reference field fluctuates it changes direction, activating other components of the slowly-evolving strain tensor. Therefore, from the perspective of an observer on a timescale much longer than the fluctuation timescale, fluctuations produce an effective undirected (although anisotropic) elastic potential energy. See the right panel in Fig.\,\ref{fig:effective_elasticity}. This explains the origin and basic physics interpretation of the excess energy term. Its functional dependence on $G$ is determined by considering the one-point magnetic variance $\overline{\Sigma}_0$, which is discussed next. 

Stellarators empirically display a phenomenon known as profile stiffness \cite{beurskensIonTemperatureClamping2021,stroteichSeekingTurbulenceOptimized2022,garbetProfileStiffnessGlobal2004}, whereby increasing plasma profile gradients beyond a threshold causes the rate of change of turbulent heat flux with profile gradients to sharply increase. Several previous studies \cite{barnesCriticallyBalancedIon2011,vanwykTransitionSubcriticalTurbulence2016,adkinsAsymptoticScalingTheory2026} indicate that, for large enough profile gradients, steeper heat flux is accompanied by larger amplitude turbulent fluctuations. Assuming critical balance, asymptotic scaling laws for the amplitude of electrostatic turbulence as a function of profile gradients have been obtained \cite{barnesCriticallyBalancedIon2011,adkinsAsymptoticScalingTheory2026}. Although similar scaling laws have yet to be derived for electromagnetic fluctuations, these results suggest that the variance of turbulent magnetic fluctuations $\overline{\Sigma}$ also grows with profile gradients when those gradients exceed some threshold. These observations suggest a particular form for the excess energy.

The literature on profile stiffness generally assumes the magnetic field has nested flux surfaces and that the ``profiles" are constant on those surfaces. Supposing flux surfaces are level sets of some function $\psi$, the profiles therefore have the form $f(\psi)$, where $f$ is a single-variable function. Since $\psi = \psi_0\circ G$ for some reference flux function $\psi_0:Q_0\rightarrow\mathbb{R}$ it follows that gradients of $f(\psi)$ are proportional to $\mathbb{G}$ according to \[|\nabla f(\psi)|^2 = |f'(\psi)|^2\,\nabla\psi_0^T(G(\bm{x}))\mathbb{G}(\bm{x})\mathbb{G}^T(\bm{x})\nabla\psi_0(G(\bm{x})).\]   
This result, in conjunction with assumed validity of the profile stiffness picture above, implies that the $\overline{\Sigma}$ should be an increasing function of plasma strain. For suppose that $\|\mathbb{A} \|$ increases. Since $\| \mathbb{G}\| $ is bounded below by $\| \mathbb{A}\|^{1/2}$ the gradient norm $\|\mathbb{G} \|$ must eventually increase as well. Again invoking the assumed rapid fluctuations in the reference magnetic field, increasing $\|\mathbb{G} \|$ causes the RMS profile gradient, $\sqrt{\mathbb{E}[|\nabla f(\psi)|^2]}$, to increase. Profile stiffness then implies that the amplitude of magnetic fluctuations must increase. Thus, the magnetic variance should be an increasing function of $\| \mathbb{A}\|$. This result applies literally to nested magnetic topologies; in what follows it will be assumed the result remains valid in non-nested topologies. Detailed study of gyrokinetic fluctuation statistics in non-nested topologies remains a largely open area of research.

A simple model that embodies this rough physics picture decomposes the excess energy as the sum of a term linear in $\mathbb{A}^T\mathbb{A}$ and a homogeneous superlinear term. The linear term corresponds to the magnetic variance $\Sigma_0$ in the reference configuration. It is reasonable to assume $\Sigma_0$ is uniformly positive definite. The homogeneous superlinear term generally depends on a higher-degree reference tensor field. But under the simplifying assumption that the nonlinear term merely renormalizes the linear term, the following phenomenological model for the mean magnetic energy emerges,
\begin{gather}
    \mathsf{W}_{\text{mag}}(G) = \int \mathsf{w}_{\text{mag}}(G(\bm{x}),\mathbb{A}(\bm{x}),\mathcal{J}(\bm{x}))\,d^3\bm{x},\label{eq:phenom_mean_en}\\
    \mathsf{w}_{\text{mag}}(G,\mathbb{A},\mathcal{J}) = \frac{1}{2}\mu_0^{-1} |\bm{\mathsf{B}}|^2 + \frac{1}{2}\mu_0^{-1}\text{tr}\bigg(\Sigma + \frac{\Sigma^{\mexp}}{(\Delta B^2)^{\mexp-1}}\bigg),\label{eq:phenom_mean_en_den}\\
    \bm{\mathsf{B}} = \mathbb{A}\bm{\mathsf{B}}_0(G),\quad \Sigma = \mathbb{A}\Sigma_0(G)\mathbb{A}^T.\nonumber
\end{gather}
Here $\mexp > 1$ is an exponent that sets the homogeneous degree of the nonlinearity, $\Delta B^2$ is a characteristic magnetic variance above which the nonlinearity becomes important, the fields $\bm{\mathsf{B}}_0$, $\Sigma_0$, denote the mean reference magnetic field and the reference magnetic variance, and $\mathsf{w}_{\text{mag}}:Q_0\times\mathbb{R}^{3\times 3}\times\mathbb{R}\rightarrow \mathbb{R}\cup \{+\infty\}$ denotes the effective magnetic energy density. The parameters $\mexp$ and $\Delta B^2$ should be fit to the results of high-fidelity microturbulence simulations in various device configurations. The mean reference magnetic field $\bm{\mathsf{B}}_0$ should be determined by a suitable magnetic topology model. The reference magnetic variance $\Sigma_0$ should be determined by requiring the magnetic variance predicted by the equilibrium model agrees with the magnetic variance predicted by a gyrokinetic turbulence model.

When Newcomb's magnetic energy $\mathcal{W}_{\text{mag}}$ in  \,\eqref{eq:newcomb_mag_energy} is replaced by $\mathsf{W}_{\text{mag}}$ in  \,\eqref{eq:phenom_mean_en} the one-parameter family $G_\lambda$ introduced earlier ceases to be an escape sequence. The magnetic energy instead scales as $\lambda^{\mexp - 1}$ as $\lambda\rightarrow\infty$. Thus, including the excess energy of magnetic fluctuations simultaneously improves the physics fidelity of the equilibrium model and improves the potential energy functional's growth properties.

\section{Soft vacuum barrier\label{sec:vac_barrier}}
Including excess magnetic energy does not achieve coercivity by itself. The following argument constructs an escape sequence for Newcomb's potential with excess magnetic energy included. Let $\gamma_* = 1-\tfrac{1}{2\mexp}$ and note $0<\gamma_* < 1$. Consider a $1$-parameter family of back-to-labels maps $G_\lambda$ with parameter $\lambda\in\mathbb{R}$. Suppose that on a subregion of $Q$ with volume $V(\lambda)\propto \lambda^{-1}$ the derivative matrix is given by 
\begin{align*}
    \mathbb{G}_\lambda = \begin{pmatrix} \lambda &0&0\\ 0& \lambda^{-\gamma_*} & 0\\ 0&0&\lambda^{-1}\end{pmatrix}.
\end{align*}
Elsewhere the derivative matrix is $O(1)$. Now choose any $p>1$. The $p$-norm of the derivative matrix scales with $\lambda$ according to
\begin{align*}
\int_Q \| \mathbb{G}_\lambda\|^p\,d^3\bm{x}\propto \lambda^{p-1}.
\end{align*}
In particular, the $p$-norm of the gradient tends to $\infty$ as $\lambda\rightarrow\infty$. Coercivity of Newcomb's functional with excess magnetic energy included would therefore require $(\mathcal{W}_{\text{int}}+\mathsf{W}_{\text{mag}})(G_\lambda)\rightarrow\infty$ as $\lambda\rightarrow \infty$. To assess whether this actually happens, first note that the adjugate matrix and the Jacobian determinant in the region where $\mathbb{G}$ is large are given by
\begin{align*}
    \mathbb{A}_\lambda = \begin{pmatrix} \lambda^{-1-\gamma_*} &0&0\\ 0& 1 & 0\\ 0&0&\lambda^{1-\gamma_*}\end{pmatrix},\quad \mathcal{J}_\lambda = \lambda^{-\gamma_*}.
\end{align*}
By the superlinear dependence of $\mathsf{W}_{\text{mag}}(G_\lambda)$ on $\mathbb{A}_\lambda^T\mathbb{A}_\lambda$ the magnetic energy  therefore scales according to
\[\mathsf{W}_{\text{mag}}(G_\lambda)\propto \lambda^{2\mexp(1-\gamma_*)-1} = 1,\quad \lambda\rightarrow\infty.\]
Moreover, since $\mathcal{J}_\lambda\rightarrow 0$ and $G_\lambda$ is bounded the internal energy tends to $0$ as $\lambda\rightarrow\infty$. This shows that the one-parameter family $G_\lambda$ is an escape sequence for $\mathcal{W}_{\text{int}}+\mathsf{W}_{\text{mag}}$ for any $p>1$.

The origin of this escape sequence is failure of the internal energy or the modified magnetic energy to energetically penalize excessive rarefaction in thin regions. Very small densities are a well-known problem in MHD fluid modeling, especially in the high-energy-density plasma simulation community. When low densities become problematic it has become standard practice to introduce a density floor. A recent take on this idea \cite{farmer_numerical_2024} is described by Farmer et. al.. In addition, in the context of stellarator equilibrium modeling there is little evidence to suggest excessive rarefaction actually happens in reactor-relevant configurations. This all suggests that modifying the internal energy functional so that it sharply penalizes excessive rarefaction is practical, mathematically convenient, and physically harmless. 

It is in this spirit that the equilibrium model proposed in this work exchanges the standard ideal gas internal energy $\mathcal{W}_{\text{int}}$ with the modified internal energy $\mathsf{W}_{\text{int}}(G) = \int_Q \mathsf{w}_{\text{int}}(G,\mathbb{A},\mathcal{J})\,d^3\bm{x}$, where $\mathsf{w}_{\text{int}}:Q_0\times \mathbb{R}^{3\times 3}\times \mathbb{R}\rightarrow \mathbb{R}\cup\{+\infty\}$ is given by
\begin{gather}
    \mathsf{w}_{\text{int}}(G,\mathbb{A},\mathcal{J})=\rho\,\mathcal{U}^*(\rho,s) =\begin{cases} \frac{\overline{p}(s)}{\overline{\rho}^\gamma}\frac{\rho^{\gamma}}{\gamma-1} + \frac{\varepsilon_{\text{vac}}\,\rho_{\text{vac}}^{\vexp}}{\rho^{\vexp}}, & \rho > 0\\ +\infty & \rho\leq 0.\end{cases}\label{modified_int_en_den}\\
    \rho=\mathcal{J}\,\rho_0(G), \quad s = s_0(G)\nonumber. 
\end{gather}
Here $\gamma > 1$ denotes the adiabatic index, $\overline{p}(s)\propto \exp(s/C_v)$ denotes a characteristic entropy-dependent pressure, and $\overline{\rho}$ denotes a (constant, positive) characteristic mass density. The positive constants $\varepsilon_{\text{vac}}$ and $\rho_{\text{vac}}$ control the strength and mass density scale for the vacuum barrier; decreasing mass density below $\rho_{\text{vac}}$ incurs an energy cost proportional to $\varepsilon_{\text{vac}}$. The barrier abruptness is controlled by the exponent $\vexp$. Eventually there will be a further constraint on $\vexp$, but for now it is sufficient to assume $\vexp > 0$. Note that $\mathsf{w}_{\text{int}}$ is manifestly non-negative. Also note that the vacuum barrier term $\propto \rho^{-\vexp}$ ensures $G_\lambda$ above is no longer an escape sequence.

It will become clear in Section \ref{sec:coercivity} that the pair of modifications to Newcomb's potential, $\mathcal{W}_{\text{int}}\rightarrow\mathsf{W}_{\text{int}}$ and $\mathcal{W}_{\text{mag}}\rightarrow \mathsf{W}_{\text{mag}}$, together produce a coercive potential energy functional.

\section{Model formulation\label{sec:model_formulation}}

Let $Q\subset\mathbb{R}^3$ denote a compact region in $\mathbb{R}^3$ with $C^\infty$ boundary. Denote the volume of $Q$ as $L^3$, which gives $L$ the interpretation of a characteristic domain length scale. Let $W^{1,p}(Q,\mathbb{R}^3)$ denote the (real) linear space of functions $G:Q\rightarrow \mathbb{R}^3$ such that the $p$-norm of $G$ and its first derivative are each finite. Eventually the analysis will select a specific $p$-value. For now it suffices to assume $3 < p <\infty$, which ensures in particular that elements of $W^{1,p}(Q,\mathbb{R}^3)$ have continuous representatives. If $\mathbb{G}(\bm{x}) = DG(\bm{x})$ denotes the $3\times 3$ matrix whose columns equal the partial derivatives of $G$ at $\bm{x}\in Q$, $G\in W^{1,p}(Q,\mathbb{R}^3)$ is equivalent to $\int_Q | G(\bm{x})|^p\,d^3\bm{x} < \infty$ and $\int_Q \| \mathbb{G}(\bm{x})\|^p\,d^3\bm{x} <\infty$,
where $ | \cdot |$ denotes the standard Euclidean norm on $\mathbb{R}^3$ and $\|\cdot\|$ denotes the Frobenius norm on $3\times 3$ matrices. For the allowed range of $p$, it is a well-known fact that $W^{1,p}(Q,\mathbb{R}^3)$ is a reflexive Banach space. Of the many possible equivalent norms on $W^{1,p}(Q,\mathbb{R}^3)$ this work uses
\begin{align*}
    \| G\|_{1,p} = \bigg(\frac{1}{L^{3 + p}}\int_Q |G(\bm{x})|^p\,d^3\bm{x} + \frac{1}{L^3}\int_Q\|\mathbb{G}(\bm{x}) \|^p\,d^3\bm{x}\bigg)^{1/p},
\end{align*}
which has the benefit of being dimensionless. Let $Q_0\subset \mathbb{R}^3$ denote another region in $\mathbb{R}^3$ $C^\infty$-diffeomorphic to $Q$. Let $\mathcal{Z}\subset W^{1,p}(Q,\mathbb{R}^3)$ be the subset of $W^{1,p}(Q,\mathbb{R}^3)$ comprising maps $G\in W^{1,p}(Q,\mathbb{R}^3)$ that satisfy the three conditions
\begin{description}
    \item[BL1:] $G(Q)\subset Q_0$
    \item[BL2:] $G(\partial Q)\subset \partial Q_0$
    \item[BL3:] $\int_Q \text{det}\,\mathbb{G}(\bm{x})\,d^3\bm{x} \leq \int_{G(Q)}d^3\bm{x}_0$.
\end{description}
Note that BL3 is known as the Ciarlet-Ne\v{c}as \cite{ciarletInjectivitySelfcontactNonlinear1987} condition in the field of nonlinear elasticity.  Also note that the upper bound in $\textbf{BL3}$ exists because $G$ is continuous for $p > 3$, which implies $G(Q)$ is measurable.
The set $\mathcal{Z}$ serves as a viable fluid configuration space for the plasma. Note that $G$ in the complement of $\mathcal{Z}$ are not valid back-to-labels maps because all fluid parcel labels must be contained in the reference spatial domain $Q_0$.

\textbf{Elasticized magnetohydrostatics --- } The potential energy is the elasticized Newcomb potential $\mathsf{W}:\mathcal{Z}\rightarrow \mathbb{R}\cup \{+\infty\}$ defined as the sum of the elasticized magnetic energy given in \eqref{eq:phenom_mean_en} and the modified internal energy given in \eqref{modified_int_en_den}.
The following provides a detailed specification of the energy functional $\mathsf{W}$ in terms of a potential energy density $\mathsf{w}$. 

Let $\text{Sym}(3)$ denote the space of symmetric $3\times 3$ matrices. Let $\rho_0:Q_0\rightarrow\mathbb{R}$, $s_0:Q_0\rightarrow \mathbb{R}$, $\bm{\mathsf{B}}_0:Q_0\rightarrow\mathbb{R}^3$, and $\Sigma_0:Q_0\rightarrow \text{Sym}(3)$ be fixed fields defined over the reference domain $Q_0$. The reference fields will always satisfy the following assumption.
\begin{assumption}\label{assumption1}
    Each of the reference fields $\rho_0,s_0,\bm{\mathsf{B}}_0,\Sigma_0$ is $C^\infty$. In addition $\bm{\mathsf{B}}_0 = \nabla\times\bm{\mathsf{A}}_0$ for some $C^\infty$ vector potential $\bm{\mathsf{A}}_0$, $\Sigma_0$ is uniformly positive definite, and $\rho_0$ is positive pointwise. 
\end{assumption}
\begin{remark}
    The field $\bm{\mathsf{B}}_0$ need not be tangent to $\partial Q_0$, but Assumption \ref{assumption1} does allow for it. Tangency of $\bm{\mathsf{B}}_0$ to $\partial Q_0$ is a common confinement modeling assumption, but it will play no role in the rigorous existence theory.
\end{remark}

The elasticized potential energy density $\mathsf{w}:Q_0\times \mathbb{R}^{3\times 3}\times \mathbb{R}\rightarrow \mathbb{R}\cup\{+\infty\}$ is given by
\begin{align*}
    \mathsf{w}(G,\mathbb{A},\mathcal{J}) = \mathsf{w}_{\text{int}}(G,\mathbb{A},\mathcal{J}) + \mathsf{w}_{\text{mag}}(G,\mathbb{A},\mathcal{J}), 
\end{align*}
where $\mathsf{w}_{\text{int}}$ and $\mathsf{w}_{\text{mag}}$ are defined in \eqref{modified_int_en_den} and \eqref{eq:phenom_mean_en_den}. Note that both $\mathsf{w}_{\text{int}}$ and $\mathsf{w}_{\text{mag}}$ depend on the reference fields. The corresponding elasticized potential energy is
\begin{gather}
    \mathsf{W}(G) = \int_Q\mathsf{w}(G(\bm{x}),\mathbb{A}(\bm{x}),\mathcal{J}(\bm{x}))\,d^3\bm{x} ,\label{potential_energy_defined}\\
    \mathcal{J}(\bm{x}) = \text{det}\,\mathbb{G}(\bm{x}),\quad \mathbb{A}(\bm{x}) = \text{adj}\,\mathbb{G}(\bm{x})\nonumber.
\end{gather}

Elasticized magnetohydrostatics defines equilibria as follows.
\begin{definition}\label{equilibrium_def}
    A back-to-labels map $G\in\mathcal{Z}$ is an \textbf{equilibrium} if it is a minimizer for $\mathsf{W}$. Equivalently, $G$ is an equilibrium if, for every $G'\in\mathcal{Z}$, $\mathsf{W}(G)\leq \mathsf{W}(G')$.
\end{definition}
\begin{remark}
The definition \ref{equilibrium_def} defines equilibria in the model as ground states for the potential energy. The energy functional $\mathsf{W}$ may therefore have local minima not labeled as equilibria by the definition. As stated, the results proved in this work concerning regularity of equilibria do not apply to local minima that are not global minima. Generically one expects a single global minimizer, and therefore a unique equilibrium. However, there may be special model inputs (domain boundary and reference fields) that produce multiple global minimizers.
\end{remark}

\section{The potential energy is norm-Coercive\label{sec:coercivity}}
It is clear from the definition of $\mathsf{W}$ that $\mathsf{W}(G)\geq 0$ for all $G\in\mathcal{Z}$, as would be expected for a physical energy. In fact $\mathsf{W}$ satisfies a much stronger positivity property known as coercivity when $p$ is appropriately related to the vacuum and variance exponents, $\vexp,\mexp$.

\begin{theorem}\label{thm:coercivity}
    Suppose \[3 < \frac{ \vexp \mexp}{\vexp+\mexp} = p.\] Then the functional $\mathsf{W}:\mathcal{Z}\rightarrow \mathbb{R}\cup\{+\infty\}$ is $W^{1,p}$-norm-coercive. In fact, there are $G$-independent positive constants $c_0,c_1$ such that for each $G\in\mathcal{Z}$
    \begin{align}
        \mathsf{W}(G)\geq  \frac{\| G\|_{1,p}^p - c_0}{c_1}.\label{norm_coercivity_inequality}
    \end{align}
\end{theorem}
\begin{remark}\label{rem:p_remark}
    Since $\frac{ \vexp \mexp}{\vexp+\mexp}$ is strictly less than either $\vexp$ or $\mexp$ it follows that the Theorem's hypothesis requires, in particular, $\vexp >3$ and $\mexp>3$. The vacuum barrier is artificial, so making it more narrow by increasing $\vexp$ is desirable. On the other hand, since the amplitude scaling of electromagnetic fluctuations with supercritical profile gradients remains an open question, it is possible that $\mexp > 3$ is unrealistic. This represents a genuine weakness of the rigorous results presented here. However, numerical experiments in Section \ref{sec:discussion} suggest this is a weakness in the analytical methods used to obtain the results, rather than a weakness in the model itself.
\end{remark}

\begin{proof}
Let $G\in\mathcal{Z}$ and let $\mathbb{G}$ denote its weak derivative. The coercivity inequality \eqref{norm_coercivity_inequality} is trivially satisfied if $\mathsf{W}(G) = +\infty$. So assume that $\mathsf{W}(G) < \infty$. This implies that the modified internal energy must be finite, which in turn implies that $\mathcal{J}(\bm{x}) > 0$ for almost all $\bm{x}\in Q$.

Observe that if $\mathbb{M}$ is any invertible $3\times 3$ matrix then its Frobenius norm is bounded from above according to
\begin{align}
    \| \mathbb{M}\|^2\leq \frac{1}{3}\frac{\| \text{adj}\,\mathbb{M}\|^4}{(\text{det}\,\mathbb{M})^2},\label{pointwise_adjugate_bound}
\end{align}
where $\text{adj}\,\mathbb{M}$ denotes the adjugate of $\mathbb{M}$. Let $1 < a,b <\infty$ be real numbers that satisfy $1/a + 1/b=1$. By \eqref{pointwise_adjugate_bound} and H\"older's inequality the $p$-norm of $\mathbb{G}$ satisfies
\begin{align}
    \int_Q\| \mathbb{G}\|^p\,d^3\bm{x}  &\leq 3^{-p/2} \int_Q \frac{\| \mathbb{A}\|^{2p}}{\mathcal{J}^p}\,d^3\bm{x}\nonumber\\
    &\leq 3^{-p/2}\bigg(\int_Q\frac{1}{\mathcal{J}^{ap}}\,d^3\bm{x}\bigg)^{1/a}\bigg(\int_Q\| \mathbb{A}\|^{2pb}\,d^3\bm{x}\bigg)^{1/b}.\label{proto_coercivity}
\end{align}
Now by the definition \eqref{modified_int_en_den} of the modified internal energy density and non-negativity of the modified magnetic energy density,
\begin{align*}
    \mathsf{W}(G)\geq \int_Q\frac{\varepsilon_{\text{vac}}\rho_{\text{vac}}^{\vexp}}{(\mathcal{J}\rho_0(G))^{\vexp}}\,d^3\bm{x}\geq \frac{\varepsilon_{\text{vac}}\rho_{\text{vac}}^{\vexp}}{\rho_{0,\text{max}}^{\vexp}}\int \frac{1}{\mathcal{J}^{\vexp}}\,d^3\bm{x},
\end{align*}
where $\rho_{0,\text{max}}$ denotes the maximum reference mass density. By assumed continuity and point-wise positivity of $\rho_0$, $\rho_{0,\text{max}}^{\vexp} > 0$.
This bound suggests choosing $a$ in H\"older's inequality so that $ap = \vexp$, which is permissible because then $a = \vexp / p > 1$, by hypothesis. (Notice that $\vexp\,\mexp/(\vexp+\mexp) < \vexp$.) Proceeding in this manner, the bound in \eqref{proto_coercivity} on $\mathbb{G}$'s $p$-norm becomes
\begin{align}
    \int_Q\| \mathbb{G}\|^p\,d^3\bm{x}  \leq 3^{-p/2}\frac{\rho_{0,\text{max}}^p}{\varepsilon_{\text{vac}}^{\frac{p}{\vexp}}\rho_{\text{vac}}^p}\,\mathsf{W}(G)^{\frac{p}{\vexp}}\,\bigg(\int_Q \| \mathbb{A}\|^{2\mexp}\,d^3\bm{x}\bigg)^{\frac{\vexp - p}{\vexp}}.\label{proto_coercivity_v1}
\end{align}
Now let $\sigma_{\text{min}}$ be a uniform lower bound on the eigenvalues of $\Sigma_0$. The integrand $\| \mathbb{A}\|^{2\mexp}$ relates to $\Sigma^{\mexp}$ according to
\begin{align}
    \| \mathbb{A}\|^{2\mexp} &= [\text{tr}(\mathbb{A}^T\mathbb{A})]^{\mexp}\leq \frac{1}{\sigma_{\text{min}}^{\mexp}}[\text{tr}(\mathbb{A}\Sigma_0\mathbb{A}^T)]^{\mexp}\nonumber\\
    &\leq \frac{3^{\mexp - 1}}{\sigma_{\text{min}}^{\mexp}}\text{tr}([\mathbb{A}\Sigma_0\mathbb{A}^T]^{\mexp}) = \frac{3^{\mexp - 1}}{\sigma_{\text{min}}^{\mexp}}\text{tr}(\Sigma^{\mexp}).\label{A_to_sigma_bound}
\end{align}
Here the first inequality follows from uniform positive definiteness of $\Sigma_0$, while the second inequality follows from Jensen's inequality applied to the convex function $f(x) = x^{\mexp}$. From \eqref{A_to_sigma_bound} it follows that
\begin{align}
    \mathsf{W}(G)&\geq \frac{1}{2}\mu_0^{-1}\int_Q \frac{\text{tr}(\Sigma^{\mexp})}{(\Delta B^2)^{\mexp -1}}\,d^3\bm{x}\nonumber\\
    &\geq\frac{\mu_0^{-1}}{2\cdot 3^{\mexp-1}}\frac{\sigma_{\text{min}}^{\mexp}}{(\Delta B^2)^{\mexp - 1}}\int_Q \| \mathbb{A} \|^{2\mexp}\,d^3\bm{x}.\label{mag_en_bd}
\end{align}
Now substituting \eqref{mag_en_bd} into \eqref{proto_coercivity_v1} implies
\begin{gather}
    \int_Q \| \mathbb{G}\|^p\,d^3\bm{x} \leq 3^{p/2}\,\bigg(\frac{\rho_{0,\text{max}}}{\rho_{\text{vac}}}\frac{\Delta B^2}{\sigma_{\text{min}}}\bigg)^p\,\frac{1}{\left(\varepsilon_{\text{vac}}\right)^{p/\vexp}}\frac{1}{\left(\frac{3}{2}\mu_0^{-1}\,\Delta B^2\right)^{\frac{\vexp-p}{\vexp}}}\,\mathsf{W}(G).\label{concrete_gradient_bound}
\end{gather}

To finish the proof, note that boundedness of the reference spatial domain $Q_0$ and $G(Q)\subset Q_0$ (property $\textbf{BL1}$) implies that there is some radius $R_0$ such that $|G(\bm{x})| < R_0$ for all $\bm{x}\in Q$. Therefore $\int_Q | G(\bm{x})|^p\,d^3\bm{x}\leq R_0^p\,L^3$. In conjunction with \eqref{concrete_gradient_bound} this implies
\begin{align*}
    \| G\|_{1,p}^p &= \frac{1}{L^{3 + p}}\int_Q |G(\bm{x})|^p\,d^3\bm{x} + \frac{1}{L^3}\int_Q\|\mathbb{G}(\bm{x}) \|^p\,d^3\bm{x}\\
    &\leq \frac{R_0^p}{L^p} + \frac{3^{p/2}}{L^3}\,\bigg(\frac{\rho_{0,\text{max}}}{\rho_{\text{vac}}}\frac{\Delta B^2}{\sigma_{\text{min}}}\bigg)^p\,\frac{1}{\left(\varepsilon_{\text{vac}}\right)^{p/\vexp}}\frac{1}{\left(\frac{3}{2}\mu_0^{-1}\,\Delta B^2\right)^{\frac{\vexp-p}{\vexp}}}\,\mathsf{W}(G)
\end{align*}
which is the desired result. In particular, the positive constants $c_0,c_1$ are given by
\begin{align}
    c_0 &= \frac{R_0^p}{L^p}\label{c0_eqn}\\
    c_1 & = \frac{3^{p/2}}{L^3}\,\bigg(\frac{\rho_{0,\text{max}}}{\rho_{\text{vac}}}\frac{\Delta B^2}{\sigma_{\text{min}}}\bigg)^p\,\frac{1}{\left(\varepsilon_{\text{vac}}\right)^{p/\vexp}}\frac{1}{\left(\frac{3}{2}\mu_0^{-1}\,\Delta B^2\right)^{\frac{\vexp-p}{\vexp}}}.\label{c1_eqn}
\end{align}

\end{proof}

\section{Existence theory\label{sec:existence_theory}}
With coercivity of the elasticized potential established, existence of minimizers now follows from standard arguments from mathematical elasticity. First it is necessary to establish the following Lemmas.
\begin{lemma}\label{closed_thm}
Assume $p>3$. The optimization domain $\mathcal{Z}$ is a weakly sequentially closed subset of $W^{1,p}(Q,\mathbb{R}^3)$.
\end{lemma}

\begin{proof}
    Let $G_n\in W^{1,p}(Q,\mathbb{R}^3)$ be a weakly-convergent sequence such that $G_n\in\mathcal{Z}$ for all $n$. Since we always assume $p > 3$ and $Q$ is closed the Rellich-Kondrachov theorem implies that $G_n$ converges uniformly to a continuous map $G:Q\rightarrow\mathbb{R}^3$. We must show that the limit map $G$ satisfies properties $\textbf{BL1}$-$\textbf{BL3}$.
\\ \\
    \noindent\textbf{BL1}: Let $\bm{x}\in Q$. The sequence $G_n(\bm{x})$ is contained in $Q_0$. By uniform convergence of $G_n$ to $G$ the sequence $G_n(\bm{x})$ converges to $G(\bm{x})$. By closedness of $Q_0$ the limit of $G_n(\bm{x})$, which is $G(\bm{x})$, is contained in $Q_0$. Thus $G(Q)\subset Q_0$.
    \\ \\
    \noindent\textbf{BL2:} Let $\bm{x}\in \partial Q$. The sequence $G_n(\bm{x})$ is contained in $\partial Q_0$. By uniform convergence of $G_n$ to $G$ the sequence $G_n(\bm{x})$ converges to $G(\bm{x})$. By closedness of $\partial Q_0$ the limit of $G_n(\bm{x})$, which is $G(\bm{x})$, is contained in $\partial Q_0$. Thus $G(\partial Q)\subset \partial Q_0$.
    \\ \\
    \noindent \textbf{BL3:} We recall an argument due to Ciarlet-Ne\v{c}as \cite{ciarletInjectivitySelfcontactNonlinear1987} from the proof of their Theorem 5, which we will recount in some detail. We do this to clarify that we need not assume $\mathcal{J}_n$ is non-negative almost everywhere, which is not explicitly stated by Ciarlet-Ne\v{c}as. 
    
    For $\delta > 0$ let $B_\delta(\bm{x})$ denote the open ball with radius $\delta$ at $\bm{x}\in\mathbb{R}^3$. Since $G$ is continuous and $Q$ is compact, $G(Q)$ is compact, and hence measurable. It follows that for every $\epsilon >0$ there is an open set $O_\epsilon$ with $G(Q)\subset O_\epsilon$ and $\int_{O_\epsilon - G(Q)}\,d^3\bm{x} < \epsilon$.  In fact there is a $\delta >0$ such that $\cup_{\bm{x}\in G(Q)}B_\delta(\bm{x})\subset O_\epsilon$. Uniform convergence of $G_n$ then implies that $G_n(Q)\subset O_\epsilon$ for all sufficiently large $n$. Combining this with $\int_{Q}\mathcal{J}_n\,d^3\bm{x}\leq \int_{G_n(Q)}d^3\bm{x}_0$ then gives $\int_Q\mathcal{J}_n\,d^3\bm{x}\leq \int_{O_\epsilon}d^3\bm{x}$ for all sufficiently large $n$. 
    
    Let $r = p/3.$ The next step is to use the well-known fact that the mapping $W^{1,p}(Q,\mathbb{R}^3)\rightarrow L^r(Q)$ given by $G\mapsto \mathcal{J}$ is sequentially weakly continuous when $p>3$. (This is described for example by Ball\cite{ballConvexityConditionsExistence1976} in his Corollary 6.2.2.) Thus, by weak-convergence of $G_n$ in $W^{1,p}(Q,\mathbb{R}^3)$, the sequence $\mathcal{J}_n$ converges weakly in $L^r(Q)$ to $\mathcal{J}$. In particular, since the unit constant function $1$ is a member of the continuous dual $(L^r(Q))^*$, 
    \[\int_Q \mathcal{J}\,d^3\bm{x} = \lim_{n\rightarrow\infty}\int_Q\mathcal{J}_n\,d^3\bm{x}\leq \int_{O_\epsilon}\,d^3\bm{x}_0 = \int_{G(Q)}d^3\bm{x}_0 + \int_{O_\epsilon - G(Q)}d^3\bm{x}_0 <\int_{G(Q)}\,d^3\bm{x}_0 + \epsilon.\]
    Since this inequality must be satisfied for all $\epsilon >0$ we conclude that property \textbf{BL3} is satisfied by $G$, as claimed.

\end{proof}

\begin{lemma}\label{thm:wlsc}
    Assume $p>3$. If $G_n$ is a weakly-convergent sequence in $W^{1,p}(Q,\mathbb{R}^3)$ with weak limit $G$ and $G_n\in \mathcal{Z}$ for all $n$ then
    \begin{align*}
        \mathsf{W}(G)\leq \liminf_{n\rightarrow\infty} \,\mathsf{W}(G_n).
    \end{align*}
\end{lemma}

\begin{proof}
    The proof is a direct application of Theorem 6.4 from Bene\v{s}ova and Kru\v{z}\'{i}k\cite{benesova_weak_2017}. 
    
    We begin by extending the reference domain fields $\rho_0,\Sigma_0,s_0,\bm{\mathsf{B}}_0$ to all of $\mathbb{R}^3$ while preserving their constraints. By the Whitney extension theorem there is a $C^\infty$ function $\widetilde{s}_0:\mathbb{R}^3\rightarrow\mathbb{R}$ that agrees with $s_0$ on $Q$. Since $\rho_0$ is positive $\lambda_0 = \text{ln}\,\rho_0$ is a $C^\infty$ function on $Q_0$. Again by the Whitney extension theorem there is a $C^\infty$ extension $\widetilde{\lambda}_0:\mathbb{R}^3\rightarrow\mathbb{R}$. Since $\widetilde{\rho}_0 = \exp(\widetilde{\lambda}_0)$ is positive and agrees with $\rho_0$ on $Q_0$ it provides a $C^\infty$ extension of $\rho_0$ that is positive everywhere. Since $\Sigma_0$ is uniformly symmetric positive definite we can form is point-wise matrix logarithm $\Lambda_0 = \text{ln}\,\Sigma_0$. The matrix function $\Lambda_0:Q_0\rightarrow\mathbb{R}^{3\times 3}$ is $C^\infty$ and takes values in symmetric matrices. By the Whitney extension theorem applied to the entries of $\Lambda_0$ on the diagonal and above, there is a $C^\infty$ extension $\widetilde{\Lambda}_0:\mathbb{R}^3\rightarrow\mathbb{R}^{3\times 3}$ that also takes values in symmetric matrices. Since $\widetilde{\Sigma}_0 = \exp(\widetilde{\Lambda}_0)$ agrees with $\Sigma_0$ on $Q_0$ it provides a $C^\infty$ extension of $\Sigma_0$ to all of $\mathbb{R}^3$ that is pointwise symmetric positive definite. Finally, by application of the Whitney extension theorem component-wise to $\bm{\mathsf{A}}_0$ there is a $C^\infty$ vector field $\widetilde{\bm{\mathsf{A}}}_0:\mathbb{R}^3\rightarrow\mathbb{R}^3$ that extends $\bm{\mathsf{A}}_0$. Since $\widetilde{\bm{\mathsf{B}}}_0 = \nabla \times\widetilde{\bm{\mathsf{A}}}_0$ agrees with $\bm{\mathsf{B}}_0$ on $Q_0$ it provides a divergence-free $C^\infty $ extension of $\bm{\mathsf{B}}_0$ to all  of $\mathbb{R}^3$. Note that the extension $\widetilde{\Sigma}_0$ need not be uniformly positive definite. Also note that the extension $\widetilde{\bm{\mathsf{B}}}_0$ may have nulls, even if $\bm{\mathsf{B}}_0$ does not.

    Next we extend the domain of definition of the potential energy density. As highlighted in Section \ref{sec:model_formulation}, the potential energy density is given by a function $\mathsf{w}:Q_0\times\mathbb{R}^{3\times 3}\times\mathbb{R}\rightarrow\mathbb{R}\times\{\infty\}$. We decompose this function as a sum of four terms, $\mathsf{w} = \sum_{i=1}^4\mathsf{w}_i$, where
    \begin{align*}
        \mathsf{w}_1(G,\mathbb{A},\mathcal{J}) & = \begin{cases}\frac{\overline{p}(s_0(G))}{\overline{\rho}^\gamma}\frac{(\mathcal{J}\rho_0(G))^\gamma}{\gamma-1}+\frac{\varepsilon_{\text{vac}}\rho_{\text{vac}}^{\vexp}}{(\mathcal{J}\rho_0(G))^{\vexp}}, & \mathcal{J} >0\\ +\infty,& \mathcal{J}\leq 0 \end{cases}\\
        \mathsf{w}_2(G,\mathbb{A},\mathcal{J}) & = \frac{1}{2}\mu_0^{-1}|\mathbb{A}\bm{\mathsf{B}}_0(G)|^2\\
        \mathsf{w}_3(G,\mathbb{A},\mathcal{J}) & = \frac{1}{2}\mu_0^{-1}\text{tr}(\mathbb{A}\Sigma_0(G)\mathbb{A}^T)\\
            \mathsf{w}_4(G,\mathbb{A},\mathcal{J}) & =\frac{1}{2}\mu_0^{-1}\frac{\text{tr}([\mathbb{A}\Sigma_0(G)\mathbb{A}^T]^{\mexp})}{(\Delta B^2)^{\mexp - 1}}.
    \end{align*}
    By replacing the $C^\infty$ reference fields $\rho_0,s_0,\Sigma_0,\bm{\mathsf{B}}_0$ with their $C^\infty$ extensions $\widetilde{\rho}_0,\widetilde{s}_0,\widetilde{\Sigma}_0,\widetilde{\bm{\mathsf{B}}}_0$, we obtain extensions $\widetilde{\mathsf{w}}_i:\mathbb{R}^3\times\mathbb{R}^{3\times 3}\times\mathbb{R}\rightarrow\mathbb{R}^3$ given explicitly by
    \begin{align*}
        \widetilde{\mathsf{w}}_1(G,\mathbb{A},\mathcal{J}) & = \begin{cases}\frac{\overline{p}(\widetilde{s}_0(G))}{\overline{\rho}^\gamma}\frac{(\mathcal{J}\widetilde{\rho}_0(G))^\gamma}{\gamma-1}+\frac{\varepsilon_{\text{vac}}\rho_{\text{vac}}^{\vexp}}{(\mathcal{J}\widetilde{\rho}_0(G))^{\vexp}}, & \mathcal{J} >0\\ +\infty,& \mathcal{J}\leq 0 \end{cases}\\
        \widetilde{\mathsf{w}}_2(G,\mathbb{A},\mathcal{J}) & = \frac{1}{2}\mu_0^{-1}|\mathbb{A}\widetilde{\bm{\mathsf{B}}}_0(G)|^2\\
        \widetilde{\mathsf{w}}_3(G,\mathbb{A},\mathcal{J}) & = \frac{1}{2}\mu_0^{-1}\text{tr}(\mathbb{A}\widetilde{\Sigma}_0(G)\mathbb{A}^T)\\
            \widetilde{\mathsf{w}}_4(G,\mathbb{A},\mathcal{J}) & =\frac{1}{2}\mu_0^{-1}\frac{\text{tr}([\mathbb{A}\widetilde{\Sigma}_0(G)\mathbb{A}^T]^{\mexp})}{(\Delta B^2)^{\mexp - 1}}.
    \end{align*}
    Notice that each of the extensions is continuous and non-negative when the extended real line $\mathbb{R}\cup \{+\infty\}$ is equipped with its usual topology. 

    Now we establish convexity of $\widetilde{\mathsf{w}}_i(G,\mathbb{A},\mathcal{J})$ in $(\mathbb{A},\mathcal{J})$ for each fixed $G\in\mathbb{R}^3$. The function $\widetilde{\mathsf{w}}_1$ is convex in $(\mathbb{A},\mathcal{J})$ because it is equal to a convex single-variable function of $\mathcal{J}$. Note that this observation makes use of positivity of the extension $\widetilde{\rho}_0$. For fixed $G\in\mathbb{R}$, both $\widetilde{\mathsf{w}}_2$ and $\widetilde{\mathsf{w}}_3$ have the form $\mathsf{w}_S(\mathbb{A}) = \text{tr}(\mathbb{A}S\mathbb{A}^T)$, where $S$ is symmetric positive semidefinite. Note this uses pointwise positive definiteness of $\widetilde{\Sigma}_0$. There is a unique symmetric positive-semidefinite matrix $S^{1/2}$ such that $S = S^{1/2}S^{1/2}$. Therefore $\mathsf{w}_S(\mathbb{A}) = \text{tr}(\mathbb{A}S^{1/2} [\mathbb{A}S^{1/2}]^T) = \|\mathbb{A}S^{1/2}\|^2$. This shows the function $\mathsf{w}_S$ is given by compositin a linear map with a squared norm. Such compositions are known to be convex, which implies $\widetilde{\mathsf{w}}_i$ is convex in $(\mathbb{A},\mathcal{J})$ for $i = 2,3$. Finally convexity of $\widetilde{\mathsf{w}}_4$ follows from writing
    \begin{align*}
        \text{tr}([\mathbb{A}\widetilde{\Sigma}_0\mathbb{A}^T]^{\mexp}) = \text{tr}([(\mathbb{A}\widetilde{\Sigma}_0^{1/2})(\mathbb{A}\widetilde{\Sigma}_0^{1/2})^T]^{\mexp}) = \|\mathbb{A}\widetilde{\Sigma}_0^{1/2} \|_{2\mexp}^{2\mexp},
    \end{align*}
    where $\|\cdot \|_{2\mexp}$ denotes the Schatten $2\mexp$-norm. This implies convexity in $\mathbb{A}$ because if $\|\cdot \|_*$ is any norm on $3\times 3$ matrices then $\| \cdot \|_*^r$ is convex whenever $r > 1$; in this case $r = 2\mexp > 2$ because $\mexp > 1$.

    We now apply Theorem 6.4 from Bene\v{s}ova and Kru\v{z}\'{i}k\cite{benesova_weak_2017} as follows. Define $\Phi:Q\times \mathbb{R}^3\times (\mathbb{R}^{3\times 3}\times\mathbb{R})\rightarrow \mathbb{R}\cup\{+\infty\}$ according to
    \begin{align*}
        \Phi(\bm{x},G,\mathbb{A},\mathcal{J}) = \sum_{i=1}^4 \widetilde{\mathsf{w}}_i(G,\mathbb{A},\mathcal{J}).
    \end{align*}
    Since $\Phi$ is independent of $\bm{x}$ it is clear that $\Phi(\cdot,G,\mathbb{A},\mathcal{J})$ is measurable for each $(G,\mathbb{A},\mathcal{J})$. Since each of the $\widetilde{\mathsf{w}}_i$ is continuous it is clear that $\Phi(\bm{x},\cdot,\cdot,\cdot)$ is continuous for all $\bm{x}\in Q$. Since each of the $\widetilde{\mathsf{w}}_i$ is convex in $(\mathbb{A},\mathcal{J})$ it is also clear that $\Phi(\bm{x},G,\cdot,\cdot)$ is convex for each $(\bm{x},G)$. In addition, since each of the $\widetilde{\mathsf{w}}_i$ is non-negative the function $\Phi$ is bounded below by the zero function, which is a member of $L^1(Q)$. It follows that $\Phi$ satisfies each of the hypothesis specified in Theorem 6.4 from Bene\v{s}ova and Kru\v{z}\'{i}k. Because $p > 3$ we know that $G_n$ is contained in $L^1(Q,\mathbb{R}^3)$ and converges uniformly to $G$. By Ball's Corollary 6.2.2.\cite{ballConvexityConditionsExistence1976}, $\mathbb{A}_n$ converges weakly to $\mathbb{A}$ in $L^{p/2}(Q,\mathbb{R}^3)$ and $\mathcal{J}_n$ converges weakly to $\mathcal{J}$ in $L^{p/3}(Q,\mathbb{R}^3)$. In particular $(\mathbb{A}_n,\mathcal{J}_n)$ converges weakly to $(\mathbb{A},\mathcal{J})$ in $L^1$. Setting $z_n = G_n$ and $a_n = (\mathbb{A}_n,\mathcal{J}_n)$ in Theorem 6.4, we therefore infer
    \begin{align}
        \int_Q\Phi(\bm{x},G(\bm{x}),\mathbb{A}(\bm{x}),\mathcal{J}(\bm{x}))\,d^3\bm{x}\leq \liminf_{n\rightarrow\infty}\int \int_Q\Phi(\bm{x},G_n(\bm{x}),\mathbb{A}_n(\bm{x}),\mathcal{J}_n(\bm{x}))\,d^3\bm{x}.\label{BK_thm}
    \end{align}

    To complete the proof we note that $G_n\in\mathcal{Z}$ by hypothesis and $G\in\mathcal{Z}$ by Lemma \ref{closed_thm}. In light of the fact that $\widetilde{\mathsf{w}}_i:\mathbb{R}^3\times\mathbb{R}^{3\times 3}\times\mathbb{R}\rightarrow\mathbb{R}\cup \{+\infty\}$ extends $\mathsf{w}_i:Q_0\times\mathbb{R}^{3\times 3}\times\mathbb{R}\rightarrow\mathbb{R}\cup \{+\infty\}$, the following equalities therefore hold for almost all $\bm{x}\in Q$:
    \begin{gather*}
        \Phi(\bm{x},G(\bm{x}),\mathbb{A}(\bm{x}),\mathcal{J}(\bm{x})) = \mathsf{w}(G(\bm{x}),\mathbb{A}(\bm{x}),\mathcal{J}(\bm{x})),\nonumber\\
         \Phi(\bm{x},G_n(\bm{x}),\mathbb{A}_n(\bm{x}),\mathcal{J}_n(\bm{x})) = \mathsf{w}(G_n(\bm{x}),\mathbb{A}_n(\bm{x}),\mathcal{J}_n(\bm{x})).
    \end{gather*}
    We conclude, by  \,\eqref{BK_thm},
    \begin{align}
        \mathsf{W}(G) &= \int_Q\mathsf{w}(G(\bm{x}),\mathbb{A}(\bm{x}),\mathcal{J}(\bm{x}))\,d^3\bm{x} \nonumber\\
        &=\int_Q\Phi(\bm{x},G(\bm{x}),\mathbb{A}(\bm{x}),\mathcal{J}(\bm{x}))\,d^3\bm{x}\nonumber\\
        &\leq\liminf_{n\rightarrow\infty}\int \int_Q\Phi(\bm{x},G_n(\bm{x}),\mathbb{A}_n(\bm{x}),\mathcal{J}_n(\bm{x}))\,d^3\bm{x}\nonumber\\
        &=\liminf_{n\rightarrow\infty}\int_Q \mathsf{w}(G_n(\bm{x}),\mathbb{A}_n(\bm{x}),\mathcal{J}_n(\bm{x}))\,d^3\bm{x}\nonumber\\
        & =\liminf_{n\rightarrow\infty}\,\mathsf{W}(G_n),
    \end{align}
    as claimed.
\end{proof}

\noindent The main existence result for elasticized magnetohydrostatics is then.
\begin{theorem}\label{minimizer_existence}
Suppose \[3 < \frac{ \vexp \mexp}{\vexp+\mexp} = p.\] Then there is a $G_*\in \mathcal{Z}$ such that, for all $G\in\mathcal{Z}$, $\mathsf{W}(G_*)\leq \mathsf{W}(G)$.
\end{theorem}
\begin{proof}
    Since $\mathsf{W}(G)\geq 0$ for all $G\in\mathcal{Z}$ the lower bound $\mathsf{W}_{\text{min}} = \inf \{\mathsf{W}(G)\mid G\in\mathcal{Z}\}$ is finite. Choose a bounded decreasing sequence of real numbers $\mathsf{W}_N$, $N =1,2,\dots$, with $\lim_{N\rightarrow\infty}\mathsf{W}_N = \mathsf{W}_{\text{min}}$. Let $\mathcal{Z}_N = \{G\in\mathcal{Z}\mid \mathsf{W}(G)\leq \mathsf{W}_N\}$ denote the sublevel set at energy $\mathsf{W}_N$. For each $N$ the set $\mathcal{Z}_N$ must be non-empty, else $\mathsf{W}_N$ would be a lower bound for $\mathsf{W}$ greater than $\mathsf{W}_{\text{min}}$. Therefore there is an infinite sequence $G_N\in\mathcal{Z}$ with $G_N\in\mathcal{Z}_N$ for each $N$. By Theorem \ref{thm:coercivity} and $G_N\in\mathcal{Z}_N$, the $W^{1,p}$-norm of $G_N$ satisfies
    \begin{align}
        \| G_N\|_{1,p}\leq  \bigg(c_1\,\mathsf{W}(G_N)+c_0\bigg)^{1/p}\leq \bigg(c_1\,\mathsf{W}_N+c_0\bigg)^{1/p}.\label{coercivity_bound}
    \end{align}
    Here $c_0,c_1$ are some positive real numbers given in Eqs.\,\eqref{c0_eqn} and \eqref{c1_eqn}. Since the sequence $\mathsf{W}_N$ is bounded the inequality \eqref{coercivity_bound} implies that the sequence $G_N$ is bounded in $W^{1,p}$. 

    It is well-known that $W^{1,p}(Q,\mathbb{R}^3)$ is reflexive for $p>1$. Since $p >3$ by assumption, reflexivity of $W^{1,p}(Q,\mathbb{R}^3)$ comes for free. The Banach-Alaoglu theorem and the Eberlein-\v{S}mulian theorem therefore imply that the bounded sequence $G_N$ admits a weakly-convergent subsequence $ G_{N(n)}\in\mathcal{Z}$, $n=1,2,\dots$. Let $G_* = \lim_{n\rightarrow\infty}G_{N(n)}$. Since $\mathsf{W}_N$ is a decreasing convergent sequence, $\lim_{n\rightarrow\infty}\mathsf{W}(G_{N(n)}) = \mathsf{W}_{\text{min}}$. Moreover, by Lemma \ref{thm:wlsc},
    \begin{align*}
        \mathsf{W}(G_*)\leq \liminf_{n\rightarrow\infty}\mathsf{W}(G_{N(n)}) = \mathsf{W}_{\text{min}}.
    \end{align*}
    But $\mathsf{W}_{\text{min}}$ is a lower bound on $\mathsf{W}$, which implies $\mathsf{W}_{\text{min}}\leq \mathsf{W}(G_*)$. Thus, $\mathsf{W}(G_*) = \mathsf{W}_{\text{min}}$, and $G_*\in\mathcal{Z}$ is a global minimizer for $\mathsf{W}$, as claimed. 
\end{proof}

\section{Discussion}\label{sec:discussion}
The proof presented here for existence of minimizers of the elasticized Newcomb potential $\mathsf{W}$ largely follows a pattern established in the context of mathematical elasticity \cite{ballOpenProblemsElasticity2002a,ballConvexityConditionsExistence1976,ballMinimizersEulerlagrangeEquations1984,ciarletInjectivitySelfcontactNonlinear1987}. 
Its most novel mathematical ingredient comprises the method used to establish coercivity, which entails bounding the gradient of the back-to-labels map by its adjugate. The key inequality is given in  \,\eqref{pointwise_adjugate_bound}, which is an elementary result in linear algebra.

The idea of modifying Newcomb's potential by including the excess magnetic energy associated with turbulent fluctuations originated in earlier work in Ref.\cite{ruthStatisticalEquilibriumModel2026}. However, this earlier statistical equilibrium model assumed the one-point reference magnetic variance does not depend on the large-scale fluid state $G$. It also fell short of establishing a rigorous existence theory. In contrast, the elasticized magnetohydrostatic model presented here hypothesizes that the reference magnetic variance grows with plasma strain due to profile stiffness. This profile stiffness effect, together with the vacuum barrier, are the crucial new ingredients needed for the existence theory presented here. 


It is presently unclear whether solutions of elasticized magnetohydrostatics are always smoother than Bruno-Laurence solutions. While the rigorous existence theory presented here only guarantees minimizers in $W^{1,p}(Q,\mathbb{R}^3)$, a space that allows for jump discontinuities in the magnetic field, we find that well-resolved numerical solutions, such as those shown in Fig.\,\ref{fig:side_by_side}, do not exhibit jump discontinuities. Like the existence theory, the difficulty in assessing regularity of minimizers for the elasticized Newcomb potential is largely inherited from the theory of mathematical elasticity, where equilibrium regularity is an open problem in general \cite{ballOpenProblemsElasticity2002a}. This obstacle might be sidestepped by allowing for strain gradient dependence in the mean magnetic energy, which is specifically neglected in the present work. A plausible source of strain gradient dependence is large-scale turbulent fluid motions, as discussed in Holm's formulation of the Lagrangian-averaged $\alpha$-model \cite{holmLagrangianAveragesAveraged2002}.


\begin{figure}
    \centering
    \includegraphics[width=0.4\linewidth]{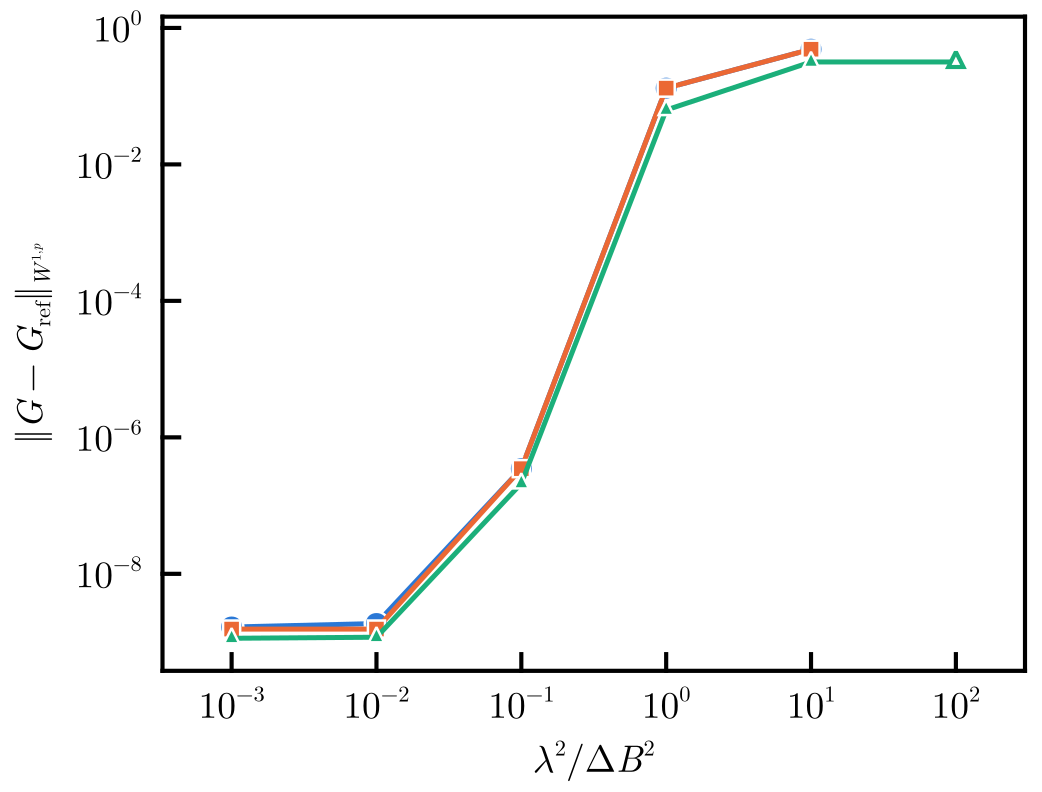}
    \includegraphics[width=0.4\linewidth]{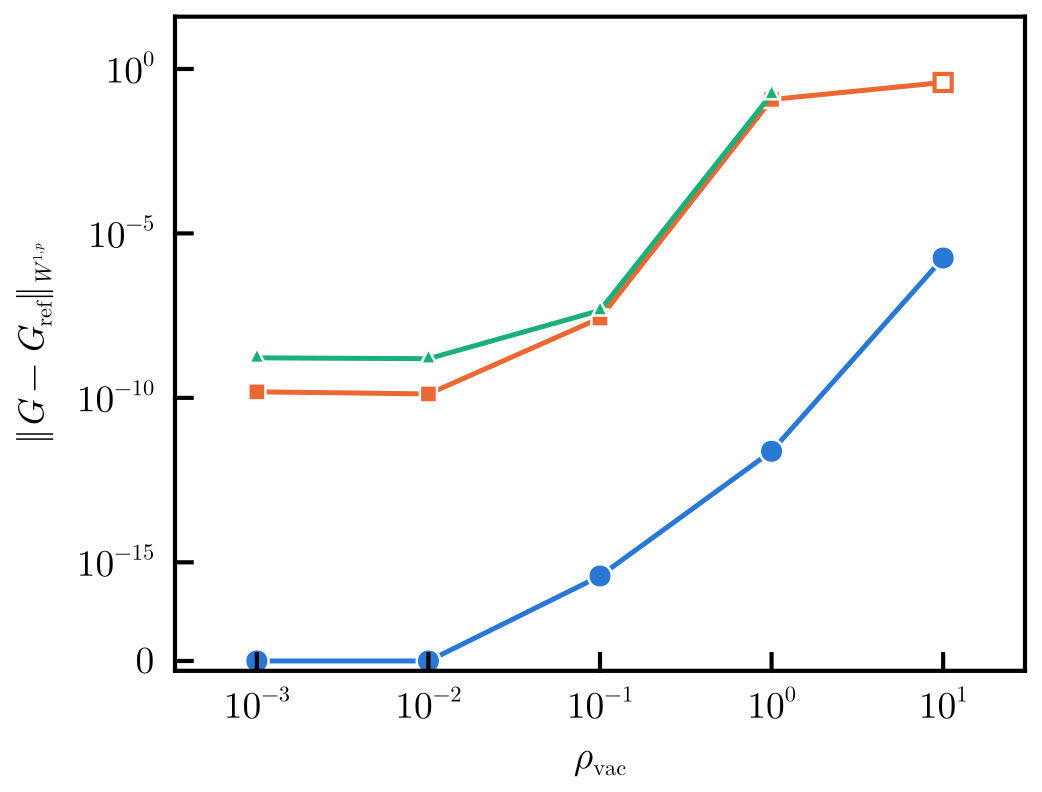}
    \caption{Self-convergence of computed solutions as either the nonlinear elastic response (left) or vacuum barrier (right) is removed. The reference magnetic variance is $\Sigma_0 = \lambda^2\,\mathbb{I}$ with $\lambda = 0.05$. Plotted curves show $W^{1,p}$-distance to a reference solution with $\Delta B^2 = 10^4\times \lambda^2$ (left) or $\rho_{\text{vac}}=0$ (right). Convergence studies with various values of $\rho_{\text{vac}}$ (left) or $\lambda^2/\Delta B^2$ (right) are indicated with different colors.}
    \label{fig:scans}
\end{figure}

The existence theory established here relies on both the vacuum barrier, proportional to $\rho_{\text{vac}}>0$ , and the nonlinear elastic response, proportional to $1/\Delta B^2 > 0$. If $\rho_{\text{vac}}\rightarrow 0$ or $1/\Delta B^2\rightarrow 0$ the elasticized potential loses coercivity and the direct method in the calculus of variations cannot be applied. On the other hand, numerical simulations summarized in Fig. \ref{fig:scans} indicate that solutions remain well-behaved when either $\rho_{\text{vac}}\rightarrow 0$ or $1/\Delta B^2\rightarrow 0$. This suggests that neither the vacuum barrier nor the nonlinear elastic response is neccessary to obtain well-behaved equilibria using the elasticized model. This conjecture is also supported both by numerics and formal asymptotics from Ref.\,\cite{ruthStatisticalEquilibriumModel2026}. 

If plasma incompressibility is introduced as an additional hypothesis in elasticized magnetohydrostatics, the vacuum barrier is completely unnecessary. It would be interesting to establish existence of incompressible minimizers without the vacuum barrier in future work. It would also be interesting to extend the modeling principles introduced here to other equations of state.


Theorem \ref{minimizer_existence} on existence of minimizers for $\mathsf{W}$ does \emph{not} establish that minimizers satisfy the Euler-Lagrange equations. The gap between establishing existence of minimizers and establishing validity of the Euler-Lagrange equations is well-known in the calculus of variations, and in the theory of mathematical elasticity in particular \cite{ballOpenProblemsElasticity2002a,ballMinimizersEulerlagrangeEquations1984}. Part of the subtlety is already clear for minimization of the single-variable function $f(x) = x^2$, with $x\geq 1$, whose solution, $x=1$, is not a critical point for $f$. The subtleties are more intricate in infinite dimensions. Given the formal similarity of the elasticized Newcomb potential to potentials for hyperelastic solids, methods for establishing variants of the Euler-Lagrange equations developed within elasticity theory likely apply to the elasticized magnetohydrostatic model. This strategy for establishing some form of the Euler-Lagrange equations would be interesting to explore in future work. The strong-form Euler-Lagrange equations for the functional \eqref{potential_energy_defined} are given by $\nabla\cdot\mathbb{T} = 0$, where the stress tensor $\mathbb{T}$ is given by
\begin{align*}
    \mathbb{T} &= -\mathsf{p}\,\mathbb{I} + \mu_0^{-1}\bigg(\bm{\mathsf{B}}\bm{\mathsf{B}}^T - \frac{1}{2}|\bm{\mathsf{B}}|^2\mathbb{I}\bigg) + \mu_0^{-1}\bigg(\Sigma  - \frac{1}{2}\text{tr}(\Sigma)\,\mathbb{I}\bigg)\\
    & + \mu_0^{-1}\frac{1}{(\Delta B^2)^{\mexp-1}}\bigg(\mexp\,\Sigma^{\mexp} - \frac{1}{2}\text{tr}([2\mexp-1]\Sigma^{\mexp})\,\mathbb{I}\bigg),
\end{align*}
and the Eulerian fields $\mathsf{p}$, $\bm{\mathsf{B}}$, $\Sigma$ are given in terms of the back-to-labels map $G$ by 
\begin{gather*}
\mathsf{p}(\bm{x}) = \rho^2(\bm{x})\,\partial_\rho\mathcal{U}^*(\rho(\bm{x}),s(\bm{x})),\quad \rho(\bm{x}) = \mathcal{J}(\bm{x})\rho_0(G(\bm{x})),\quad s(\bm{x}) = s_0(G(\bm{x})),\\
 \bm{\mathsf{B}}(\bm{x}) = \mathbb{A}(\bm{x})\bm{\mathsf{B}}_0(G(\bm{x})),\quad \Sigma(\bm{x}) = \mathbb{A}(\bm{x})\Sigma_0(G(\bm{x}))\mathbb{A}^T(\bm{x}).
\end{gather*}
Here $\mathcal{U}^*$ is defined in \eqref{modified_int_en_den}.

\dataavail{All code used to produce the figures in this manuscript is freely available at \texttt{https://doi.org/10.5281/zenodo.23649878}.}

\acknow{This research was supported by: (I) the U.S. Department of Energy, Office of Science,
Office of Fusion Energy Sciences, under contract DE-FG02-04ER54742; (II) a Google.org AI Impact Challenge: AI for Science award; and (III) the U.S. Department of
Energy, Office of Science, Office of Advanced Scientific Computing Research, as a part of the
Mathematical Multifaceted Integrated Capability Centers program, under Award Number
DE-SC0023164. The authors used Claude Opus 4.8 and 5.5 during manuscript preparation to assist with verification of analytical calculations. All Article prose was prepared by the authors without the assistance of artificial intelligence. The authors takes full responsibility for all final content, analysis, and interpretations.  }

\showacknow 

\bibliography{pnas-sample}

\end{document}